\documentclass[12pt]{article}

\usepackage{graphicx}

\usepackage{amssymb,amsmath}
\usepackage{amsmath, amsthm}

\newtheorem{assmpn}{Assumption}[section]
\newtheorem{rem}{Remark}[section]
\usepackage{appendix}
\newtheorem{thm}{Theorem}
\newtheorem{lemma}{Lemma}
\DeclareGraphicsExtensions{.pdf,.png,.jpg}

\def\ep{\epsilon}

\def\R{{\mathbb{R}}}

\def\xept{X_{\ep,t}}

\def\yept{Y_{\ep,t}}
\def\s{\sigma}
\def\rt{R_{\ep,t}}
\def\Rt{\tilde{R}_{\ep,t}}

\def\zs{Z_{\ep,s}^{(1)}}

\def\zzs{Z_{\ep,s}^{(2)}}

\def\Zs{\tilde{Z}_{\ep,s}^{(1)}}

\def\ZZs{\tilde{Z}_{\ep,s}^{(2)}}
\def\l{\left}
\def\r{\right}
\def\p{\partial}
\def\L{\mathcal L}
\def\B{\mathcal B}
\def\Lnl{\L^{NL}_2}
\def\ul{\underline}
\def\ol{\overline}
\newcommand{\Keywords}[1]{\par\noindent 
{\small{\em Keywords\/}: #1}}

\title{Effect of volatility clustering on indifference pricing of options by convex risk measures}
\author{ Rohini Kumar  \thanks{Department of Mathematics,
      Wayne State University, Detroit, MI 48202 ({\tt rkumar@math.wayne.edu})}
      \thanks{Work  supported in part by National Science Foundation grant DMS 1209363.}
         }

\begin{document}
\maketitle

\begin{abstract}
In this paper, we look at the effect of volatility clustering on the risk indifference price of options described by Sircar and Sturm in the paper `From smile asymptotics to market risk measures'  \cite{SS}. The indifference price in \cite{SS} is obtained by using dynamic convex risk measures given by backward stochastic differential equations (BSDEs). Volatility clustering is modeled by a fast mean-reverting volatility in a stochastic volatility model for stock price. Asymptotics of the indifference price of options and their corresponding implied volatility are obtained in this paper, as the mean-reversion time approaches zero. Correction terms to the asymptotic option price and implied volatility are also obtained. \\

\end{abstract}

\section{Introduction}
In an incomplete market, there are several ways of pricing options. The arbitrage-free method of option pricing,  where the option price is given by the expected value of the discounted payoff under a risk-neutral equivalent martingale measure, is widely known and studied (see for example \cite{FPS00}). Another method is the indifference pricing method. Indifference pricing of options using utility functions have been studied extensively in \cite{HN89}, \cite{MZ04}, \cite{RKar00}, etc. Later it was observed that the same idea of indifference pricing could be extended from utility functions to dynamic convex risk measures, see \cite{BK09}.  Sircar and Sturm in \cite{SS}  derived a nonlinear partial differential equation (PDE) characterizing the indifference price of put options given by dynamic convex risk measures. In their paper, they used the residual risk measure after hedging for pricing options; the residual risk measure was given in terms of a BSDE.
They also obtained the implied volatility, corresponding to this indifference option price, as the viscosity solution of a nonlinear PDE. 

The significance of the results in \cite{SS} is that, via the indifference pricing scheme,  the market risk, reflected in the implied volatility skew, can be related to convex risk measure theory. Typically, convex risk measures are defined abstractly via BSDEs (see \cite{BK09}).  So the right driver in the BSDE which gives a good risk measure is hard to determine.  Using the indifference pricing of Sircar and Sturm, we can calibrate the driver from the market implied volatility data.  In their paper, \cite{SS}, small-maturity asymptotics of the implied volatility yielded simple formulas. However, in general, a closed form solution to the non linear PDE in\cite{SS}  is hard, if not impossible, to find. In this paper, an effort is made to give some meaningful, simplified formulas for the implied volatility surface, by considering the effect of volatility clustering.

It is believed that market volatility fluctuates frequently between high and low periods. While volatility cannot be observed directly, this ``clustering" behavior is estimated from observed stock prices.  To model stock prices subject to this clustering behavior of volatility we use stochastic volatility models where the volatility is a  fast mean-reverting ergodic process. Such stochastic volatility models with fast mean-reverting volatility were found to be a good fit for stock price data (see chapter 4 of \cite{FPS00}).  
The question of interest in this paper is, how does fast mean-reversion in volatility affect option prices? As the rate of mean-reversion increases, the long-run behavior of the ergodic volatility process manifests. Consequently,  the effect of volatility gets averaged  with respect to the invariant distribution of the ergodic volatility process.
The analysis of fast mean-reversion of volatility on option prices has been studied in the case of no-arbitrage pricing, see chapter 5 of \cite{FPS00}.  In this paper, we will look at the effect of fast mean-reversion of volatility on the indifference prices of options given in \cite{SS}.

  In the no-arbitrage pricing case, option prices are given as the solution of linear PDEs. For small mean-reversion time, denoted by the parameter $\ep$, this leads to a singular perturbation problem. In \cite{FPS00}, assuming the volatility process is an Ornstein-Uhlenbeck  process, the option price is expanded in powers of $\ep$ and the asymptotic option price, when $\ep\to0$, is obtained together with a correction term of order $\sqrt{\ep}$. The corresponding corrected implied volatility is also obtained. 
 
  In this paper, unlike the no-arbitrage case, the option price is not the solution of a linear PDE. In fact, this  indifference price is given in terms of the solutions to BSDEs. Initially, finding the asymptotic indifference price appears to involve 
averaging of BSDEs. 
 However, there is a difficulty: Sircar and Sturm use quadratic drivers in their BSDEs which are not Lipschitz and hence we cannot use the established stability results for BSDEs (for example in \cite{HPeng97}). We instead avail of the nonlinear Feynman-Kac formula derived by Pardoux and Peng in \cite{PP92},  and express the option price in terms of  solutions to nonlinear PDEs; thus converting this into a problem of averaging nonlinear PDEs.
 
 The formal derivation of the corrected option price and implied volatility, in sections \ref{sec:heuristics} and \ref{sec:IV},  follows  the same line of reasoning seen in \cite{FPS00}. The difference in method lies in the rigorous proof of accuracy of the corrected asymptotic formula, which is in section \ref{sec:rigorous}. 
 For the rigorous proof, we use the maximum principle to bound the solution of the nonlinear PDE, which gives the option price, within $o(\sqrt{\ep})$ distance of the corrected asymptotic option price formula. 
  
  The purpose of this paper is two-fold: on one hand, we look at the effect of an observed phenomenon, viz., volatility clustering, on option pricing; secondly, the averaging of the volatility leads to a  simpler  formula for implied volatility, which makes calibration easier. The dependence of the corrected implied volatility formula on the risk parameters can potentially be used to calibrate risk measures from the market implied volatility data.

It should be interesting to see if this work extends to other risk measures besides those given by BSDEs. This however falls outside the scope of this paper and will be considered for future work. In general, regardless of risk measure used,  if the security price can be expressed as a solution of a PDE, we can potentially use the method in this paper for averaging out the effect of fast mean-reverting volatility. The paper is organized as follows. In section \ref{sec:model} we recall the indifference pricing of options from \cite{SS} and introduce the nonlinear PDE that gives the indifference price of put options. Section  \ref{sec:heuristics} has heuristic calculations for the corrected asymptotic option price. The main result and the rigorous proof of accuracy of the obtained corrected asymptotic option price formula is in section \ref{sec:rigorous}.   The corrected implied volatility  formula is obtained in section \ref{sec:IV}. 


\section{Preliminaries} \label{sec:model}We begin by introducing the stochastic volatility model for stock price. 
\subsection{Stochastic volatility model.}
Let $S_{\ep,t}$ denote stock price at time $t$, where the parameter $\ep$ refers to the mean-reversion time scale  for volatility. As $\ep$ approaches $0$, the speed of mean-reversion, $1/\ep$, increases.  Let $(\Omega, \mathcal F, P)$ denote the probability space on which $S_{\ep,t}$ satisfies the following stochastic volatility model.
\begin{subequations}\label{FSDEs}
\begin{align}
dS_{\ep,t}=&b(\yept)S_{\ep,t}dt+\s_1(\yept)S_{\ep,t}dW^{(1)}_t, \quad 0\leq t\leq T,\\
d\yept=&\frac{m-\yept}{\ep}dt+\frac{\s_2(\yept)}{\sqrt{\ep}}(\rho dW_t^{(1)}+\sqrt{1-\rho^2}dW^{(2)}_t), \quad 0\leq t\leq T,\\
(S_{\ep,0},Y_{\ep,0})&=(s,y),
\end{align}
\end{subequations}
where $|\rho|< 1$, $W^{(1)}$ and $W^{(2)}$ are independent Brownian motions on $(\Omega,\mathcal F, P)$. 
We make the same assumptions on the stochastic volatility model as in \cite{SS}. For the reader's convenience we recall these assumptions.
\begin{assmpn}\label{assmpn}
We assume that 
\begin{enumerate}
\item $\s_1, \s_2\in C^{1+\beta}_{loc}(\R)$, where  $C^{1+\beta}_{loc}(\R)$ is the space of differentiable functions with locally H\"older-continuous derivatives with H\"older-exponent $\beta>0$.\\

\item Both $\s_1$ and $\s_2$ are bounded and bounded away from zero:
\[0<\underline{c_1}<\s_1<\overline{c_1}<\infty,\qquad \text{ and }0<\underline{c_2}<\s_2<\overline{c_2}<\infty,\]
\item $b\in C^{0+\beta}_{loc} $, and $b$ is bounded.
\end{enumerate}
\end{assmpn}

Let $\B$ denote the infinitesimal generator of the $Y$ process when $\ep=1$. Then, for $f\in C^2(\R)$, 
\begin{equation}\label{generator-B}
\B f(y):= (m-y)\p_yf(y)+\frac1 2 \s_2^2(y)\p^2_{yy}f(y).
\end{equation}
By the general theory of 1-D diffusions (see Karlin and Taylor~\cite{KT81}, page 221) it is easy to see that there exists a unique probability measure 
\begin{equation}\label{inv-meas}
\pi(dy)=Z^{-1}\frac{\exp\l\{\int_0^y\frac{2(m-z)}{\s_2^2(z)}dz\r\}}{\s_2^2(y)}dy,
\end{equation}
 such that $\int \B f(y) \pi(dy)=0$ for  all $f\in C_c^2(\R)$; $Z$ is the normalizing constant, so that $\int \pi(dy)=1$. The invariant distribution of $Y$ given by \eqref{inv-meas} plays an important role in the following analysis.

\subsection{Indifference option price.}
We consider a European put option with maturity time $T$ and strike price $K$. The indifference price of this European put option at time $t$, $P_{\ep}(t,x,y)$,   
is given in \cite{SS} in terms of risk measures as follows. 
\begin{equation}
P_\ep(t,x,y)=\Rt-\rt,\end{equation}
where $(\tilde{R}_\cdot, \tilde{Z}_\cdot)$ and $(R_\cdot, Z_\cdot)$ are respectively solutions of the following BSDEs
\begin{equation}\label{BSDE1}
\Rt=-\int_t^Tf(\Zs, \ZZs)ds-\int_t^T\Zs dW_s^{(1)}-\int_t^T\ZZs dW_s^{(2)};
\end{equation}
and 
\begin{equation}\label{BSDE2}
\rt=-(K-S_{\ep,T})^+-\int_t^Tf(\zs, \zzs)ds-\int_t^T\zs dW_s^{(1)}-\int_t^T\zzs dW_s^{(2)}.
\end{equation}
The function $f$  in the above BSDEs satisfies the criteria for admissible drivers (see Definition 2.3 in \cite{SS}) to ensure the solvability of the BSDEs. In this paper, we will consider a specific family of admissible drivers called {\it distorted entropic risk measures} which were introduced in \cite{SS} (see section 3.1 of \cite{SS}). 
This class of drivers has the following form:
\[g^{\eta,\gamma}(z_1,z_2):=\frac{\gamma}{2}\l((z_1+\eta z_2)^2+z_2^2\r),\] and is parametrized by two parameters: the risk aversion parameter $\gamma>0$ and the volatility risk premium $\eta$. When $\eta=0$ the driver reduces to the classical entropic risk measure whose level curves are circles with radius depending on the risk aversion parameter, $\gamma$. By introducing $\eta$, we distort this circle into an ellipse (if $|\eta|<1$). Under hedging, the risk measure gets adjusted and is now given by a BSDE where the driver $g^{\eta,\gamma}$ is transformed to 
\begin{align*}
\tilde{g}^{\eta,\gamma}(z_1,z_2)&=\inf_{\nu\in \R}\l(g^{\eta,\gamma}(z_1+\s_1(y)\nu,z_2)+b(y)\nu\r)\\
&=z_1\frac{b(y)}{\s_1(y)}+\frac{b^2(y)}{2\gamma \s^2_1(y)}-\frac{\gamma}{2}z_2^2+\frac{\eta b(y)z_2}{\s_1(y)}.\end{align*}
Henceforth we will take this to be our driver, i.e.  define
\begin{equation*}f(z_1,z_2):=z_1\frac{b(y)}{\s_1(y)}+\frac{b^2(y)}{2\gamma \s^2_1(y)}-\frac{\gamma}{2}z_2^2+\frac{\eta b(y)z_2}{\s_1(y)}.\end{equation*} 
 It is important to note that this driver is not Lipschitz, as it is quadratic in $z_2$

Instead of stock price $S_{\ep,t}$, we will work with the logarithm of  stock price normalized by the strike price: $X_{\ep,t}:=\ln \l(\frac{S_{\ep,t}}{K}\r)$. With this change of variable, \eqref{FSDEs} becomes
\begin{subequations}\label{FSDEs-alt}
\begin{align}
d\xept=&\l(b(\yept)-\frac1 2\s_1^2(\yept)\r)dt+\s_1(\yept)dW^{(1)}_t, \quad 0\leq t\leq T,\\
d\yept=&\frac{m-\yept}{\ep}dt+\frac{\s_2(\yept)}{\sqrt{\ep}}(\rho dW_t^{(1)}+\sqrt{1-\rho^2}dW^{(2)}_t), \quad 0\leq t\leq T,\\
(X_{\ep,0},Y_{\ep,0})&=(x,y).
\end{align}
\end{subequations}

Sircar and Sturm in \cite{SS} use the generalized Feynman-Kac formula, given by Pardoux and Peng in \cite{PP92},  to describe the solutions of the forward-backward SDEs, $\Rt$ and $\rt$, in terms of solutions of the nonlinear PDE \eqref{PDE} below. Before introducing this PDE, we will first make another change of variable $\tau:=T-t$, which gives the time to maturity. Define the differential operator $\L_\ep$ by
\begin{equation}\label{L_ep}
\begin{split}
\L_\ep g:=&\frac{1}{2}\s_1^2(y)\p^2_{xx}g+\frac{1}{2\ep}\s_2^2(y)\p^2_{yy}g+\frac{\rho}{\sqrt{\ep}}\s_1(y)\s_2(y)\p^2_{xy}g\\
&+\l(\frac{1}{\ep}(m-y)-\frac{\rho}{\sqrt{\ep}}\frac{b(y)\s_2(y)}{\s_1(y)}\r)\p_y g-\frac1 2 \s_1^2(y)\p_xg\\
&-\frac{b^2(y)}{2\gamma\s^2_1(y)}+\frac{\gamma(1-\rho^2)}{2\ep}\s_2^2(y)(\p_yg)^2-\frac{\eta\sqrt{1-\rho^2 }}{\sqrt{\ep}}\frac{b(y)\s_2(y)\p_y g}{\s_1(y)},
\end{split}
\end{equation} for $g\in C^{2}(\R^2)$. Let $\tilde{u}_\ep$ and $u_\ep$ denote solutions to the PDE
\begin{equation}\label{PDE}
\p_\tau u=\L_\ep u,\end{equation} with the initial conditions $\tilde{u}_\ep(0,x,y)=0$ and $u_\ep(0,x,y)=-[K-Ke^x]^+$, respectively. By Theorem 2.9 in \cite{SS}, we get the put option price
\begin{equation}\label{Pep}P_\ep(\tau,x,y)=\tilde{u}_\ep(\tau,y)-u_\ep(\tau,x,y).\end{equation}  As mentioned in \cite{SS}, by Ladyshenskaya et al.\ \cite[Theorem V.8.1]{Lad67},   $u_\ep$ and $\tilde{u}_\ep$ are unique bounded classical  solutions to the semilinear parabolic equation \eqref{PDE},  with bounded derivatives in $[0,T]\times \mathbb R\times \R$. 
\begin{rem}Since the coefficients of the PDE in \eqref{PDE} and the initial condition of $\tilde{u}$ are x-independent, we get $\tilde{u}$ to be x-independent.\end{rem}

\begin{rem}
While we only consider European put options, the results in the paper extend to any other option with bounded and continuous payoff. 
\end{rem}

\section{Asymptotic option price}\label{sec:asymptotics_option_price}
We will begin with heuristic arguments for obtaining the asymptotic option price and correction terms to the asymptotic price. The heuristic calculations follow along the same lines as the no-arbitrage option pricing case, seen in \cite{FPS00}
\subsection{Heuristics}\label{sec:heuristics}
 Assume the following expansions of $u_\ep, \tilde{u}_\ep$ and $P_\ep$ in powers of $\sqrt{\ep}$:
\begin{subequations}
\begin{align}
u_\ep &=u_0+\sqrt{\ep}u_1+\ep u_2+\ep^{3/2}u_3+\hdots, \label{u-exp}\\
\tilde{u}_\ep&=\tilde{u}_0+\sqrt{\ep}\tilde{u}_1+\ep \tilde{u}_2+\ep^{3/2}\tilde{u}_3+\hdots, \label{tildeu-exp}\\
P_\ep&=P_0+\sqrt{\ep}P_1+\ep P_2+\hdots.\label{P-exp}
\end{align}
\end{subequations}
We define the following differential operators: for $g\in C^{2}( \mathbb R\times \mathbb R)$, 
\begin{eqnarray*}
\L_0g(x,y):&=&\frac1 2 \s_1^2(y)\p^2_{xx}g(x,y)-\frac1 2 \s_1^2(y)\p_xg(x,y)-\frac{b^2(y)}{2\gamma\s_1^2(y)},\\
\L_1g(x,y):&=&\rho\s_1(y)\s_2(y)\p^2_{xy}g(x,y)\\
&&-\l(\rho+\eta\sqrt{1-\rho^2 }\r)\frac{b(y)\s_2(y)}{\s_1(y)}\p_y g(x,y),\\
\Lnl g(x,y):&=&\B g(x,y)+\frac{\gamma(1-\rho^2)}{2}\s_2^2(y)\l(\p_yg(x,y)\r)^2.
\end{eqnarray*}
Note that the operator $\Lnl$ is nonlinear.
Equation \eqref{PDE} can be rewritten as 
\begin{equation}\label{PDE-op}
\p_\tau u_\ep=\L_\ep u_\ep= \L_0u_\ep+\frac{1}{\sqrt{\ep}}\L_1u_\ep+\frac{1}{\ep}\Lnl u_\ep.
\end{equation}

\subsubsection{Leading order term}
Using the expansion \eqref{u-exp} in \eqref{PDE-op} and collecting terms of order $1/\ep$ we get
\[\B u_0= -\frac{\gamma(1-\rho^2)}{2}\s_2^2(y)(\p_yu_0)^2,\] which is satisfied if $u_0$ is $y$ independent. We will thus assume $u_0(\tau,x)$ is independent of $y$. Using this and collecting terms of order $1/\sqrt{\ep}$ we get
\[\B u_1=0,\] which is satisfied by taking $u_1$ independent of $y$. As $u_0(\tau,x)$ and $u_1(\tau,x)$ are both $y$ independent, terms of $O(1)$ in equation \eqref{PDE} satisfy
\[\p_\tau u_0(\tau,x)=\L_0 u_0(\tau,x)+\B u_2.\]
Thus, $u_2$ must satisfy the Poisson equation
\begin{equation}\label{Poisson-u2}\B u_2=\p_\tau u_0(\tau,x)-\L_0 u_0,\,\end{equation}
which has a solution provided the following centering condition holds:
\begin{equation}\label{centering}
\p_\tau u_0=\frac1 2 \overline{\s_1}^2\p^2_{xx}u_0-  \frac1 2\overline{\s_1}^2(y)\p_xu_0 - \frac{1}{2\gamma}\overline{\frac{b^2(y)}{\s_1^2(y)}}.
\end{equation}
 Here $\overline{\s_1}^2$ and $\overline{\frac{b^2(y)}{\s_1^2(y)}}$ denote the average of the terms $\s_1^2(y)$ and $\frac{b^2(y)}{\s_1^2(y)}$, respectively,  with respect to the invariant distribution, $\pi$, of the $Y$ process.
Similarly, $\tilde{u}_0$ satisfies  equation \eqref{centering}. The initial conditions for $u_\ep$ and $\tilde{u}_\ep$ give the initial conditions for $u_0$ and $\tilde{u}_0$ respectively, i.e. $\tilde{u}_0(0)=0$ ($\tilde{u}_0$ is independent of both $x$ and $y$, so only a function of $\tau$) and $u_0(0,x)=-[K-Ke^x]^+$. 
Observe that the first order approximation term to the option price,  $P_0(\tau,x)=\tilde{u}_0(\tau)-u_0(\tau,x)$, satisfies the equation \[\p_tP_0=\frac1 2 \overline{\s_1}^2\p^2_{xx}P_0-  \frac1 2\overline{\s_1}^2(y)\p_xP_0,\qquad P_0(0,x)=[K-Ke^x]^+,\]
which is simply the equation for the Black-Scholes put option price, $P_{BS}(\tau,x;\overline{\s_1})$,  with volatility parameter $\overline{\s_1}.$ Therefore,
\begin{equation}\label{P0}P_0(\tau,x)=P_{BS}(\tau,x;\overline{\s_1})=KN(-d_2)-Ke^xN(-d_1),\end{equation}where $N(z)=\frac{1}{\sqrt{2\pi}}\int_{-\infty}^z e^{-y^2/2}dy$, $d_1=\frac{x+\frac1 2 \overline{\s_1}^2\tau}{\overline{\s_1}\sqrt{\tau}}$ and 
$d_2=\frac{x-\frac1 2 \overline{\s_1}^2\tau}{\overline{\s_1}\sqrt{\tau}}$. Observe that
\begin{equation}\label{u0}u_0(\tau,x)=P_0(\tau,x)- \frac{1}{2\gamma}\overline{\frac{b^2(y)}{\s_1^2(y)}}\tau.\end{equation}
 
 \subsubsection{Higher order terms}
 Define $\hat{\L}_0:=\L_0+\frac{b^2}{2\gamma\s^2_1}=\frac1 2 \s_1^2\p^2_{xx}-\frac1 2 \s_1^2\p_x$. So $\hat{\L}_0$  is a linear operator.
 Equating terms of order $\sqrt{\ep}$ in \eqref{PDE} gives the equation
\[ \p_\tau u_1=\hat{\L}_0u_1+\L_1u_2+\B u_3.\]
Thus  $u_3$ is the solution of the  Poisson equation 
\begin{equation}\label{u3-equation}
\B u_3= \p_\tau u_1-\hat{\L}_0u_1-\L_1u_2.\end{equation}
 provided the following  centering condition is satisfied:
\[\p_\tau u_1=\overline{\hat{\L}_0u_1+\L_1u_2}.\]
Henceforth, a line over any term indicates averaging of the term with respect to the invariant distribution of the $Y$ process given in \eqref{inv-meas}.
Since $u_1$ is independent of $y$, the above equation becomes
\begin{equation}\label{u1-1}
\p_\tau u_1=\frac1 2 \overline{\s_1}^2\p^2_{xx}u_1-\frac1 2 \overline{\s_1}^2\p_xu_1+\overline{\L_1u_2}.
\end{equation}
To simplify the right hand side of the above, recall that $u_2$ is the solution of the Poisson equation in \eqref{Poisson-u2}. Together with \eqref{centering}, we see that $u_2$ is the solution of the following equation,
\[\B u_2(\tau,x,y)=\frac{b^2(y)}{2\gamma\s_1^2(y)}-\overline{\frac{b^2(y)}{2\gamma\s_1^2(y)}}+\frac1 2\l(\s_1^2(y)-\overline{\s_1}^2\r)\l(\p_{x}u_0(\tau,x)- \p^2_{xx}u_0(\tau,x)\r).\]
Let $\phi_1(y)$ and $\phi_2(y)$ denote the solutions of 
\[\B \phi_1(y)=\frac{b^2(y)}{2\gamma\s_1^2(y)}-\overline{\frac{b^2(y)}{2\gamma\s_1^2(y)}}\] and \[\B \phi_2(y)=\frac1 2 \l(\s_1^2(y)-\overline{\s_1}^2\r)\] respectively. 
Define \begin{equation}\label{U2}
U_2(\tau,x,y):= \phi_1(y)+\phi_2(y)\l(\p_x u_0(\tau,x)-\p^2_{xx}u_0(\tau,x)\r).\end{equation}
Then \begin{equation}\label{u2}
u_2(\tau,x,y)=U_2(\tau,x,y) +F(\tau, x)
\end{equation}
where $F(\tau,x)$ will be determined later, see \eqref{F}.
\begin{rem}\label{Rem-U2}
Note that the function $U_2(\tau,x,y)$ satisfies equations \eqref{Poisson-u2} and \eqref{u3-equation}.\end{rem}
We compute 
\begin{align*}
&\overline{\L_1u_2}\\
&=\overline{\rho\s_1(y)\s_2(y)\phi^\prime_2(y)}\l(\p^2_{xx}u_0(\tau,x)-\p^3_{xxx}u_0(\tau,x)\r)\\
&-\l(\rho+\eta\sqrt{1-\rho^2 }\r)\overline{\frac{b(y)\s_2(y)}{\s_1(y)}\phi^\prime_2(y)}\l(\p_x u_0(\tau,x)-\p^2_{xx}u_0(\tau,x)\r)\\
&\quad-\l(\rho+\eta\sqrt{1-\rho^2 }\r)\overline{\frac{b(y)\s_2(y)}{\s_1(y)}\phi^\prime_1(y)}.
\end{align*}
Substituting this in  \eqref{u1-1} we get
\begin{align*}
\p_\tau u_1=&\frac1 2 \overline{\s_1}^2\p^2_{xx}u_1-\frac1 2\overline{\s}_1^2(y)\p_xu_1-A\p^3_{xxx}u_0(\tau,x)\\
&+(A+B)\p^2_{xx}u_0(\tau,x)-B\p_xu_0(\tau,x)-\tilde{A},
\end{align*}
where \begin{equation}\label{A}A=\overline{\rho\s_1(y)\s_2(y)\phi^\prime_2(y)},\end{equation} \begin{equation}\label{tildeA}\tilde{A}=\l(\rho+\eta\sqrt{1-\rho^2 }\r)\overline{\frac{b(y)\s_2(y)}{\s_1(y)}\phi^\prime_1(y)}\end{equation} and \begin{equation}\label{B}B=\l(\rho+\eta\sqrt{1-\rho^2 }\r)\overline{\frac{b(y)\s_2(y)}{\s_1(y)}\phi^\prime_2(y)}.\end{equation} It is easy to verify that the solution to the above equation is
\begin{equation}\label{u1-formula}
u_1(\tau,x)=\tau\l[ -A\p^3_{xxx}u_0(\tau,x)+(A+B)\p^2_{xx}u_0(\tau,x)-B\p_xu_0(\tau,x)-\tilde{A}   \r].
\end{equation}
By a similar argument, we get 
\begin{equation*}
\tilde{u}_1(\tau)=\tau\l[ -A\p^3_{xxx}\tilde{u}_0(\tau)+(A+B)\p^2_{xx}\tilde{u}_0(\tau)-B\p_x\tilde{u}_0(\tau)-\tilde{A} \r] = -\tilde{A}\tau .
\end{equation*}
Thus, \begin{equation}\label{P1}P_1(\tau,x)=\tau\l[ -A\p^3_{xxx}P_0(\tau,x)+(A+B)\p^2_{xx}P_0(\tau,x)-B\p_xP_0(\tau,x)   \r].\end{equation}

\subsubsection{Terms of order $\ep$}\label{sec:u2}
Using the asymptotic expansion of $u_\ep$ in \eqref{PDE} and collecting terms of $O(\ep)$, we get
\[\p_\tau u_2=\hat{\L}_0u_2+\L_1 u_3+\B u_4+\frac{\gamma(1-\rho^2)}{2}\s_2^2(\p_y u_2)^2.\]
Thus $u_4$ is the solution of the Poisson equation 
\begin{equation}\label{u4}
\B u_4=\p_\tau u_2-\hat{\L}_0u_2-\L_1 u_3-\frac{\gamma(1-\rho^2)}{2}\s_2^2(\p_y u_2)^2,\end{equation} provided the following centering condition holds:
\begin{equation}\label{F}
\overline{\p_\tau u_2-\hat{\L}_0u_2-\L_1 u_3-\frac{\gamma(1-\rho^2)}{2}\s_2^2(\p_y u_2)^2}=0.\end{equation}
The unknown function $F(\tau,x)$ in \eqref{u2} is determined by the above equation \eqref{F}.

\subsection{Accuracy of corrected asymptotic price formula}\label{sec:rigorous}
\begin{thm}\label{asymp-price}
Let $P_\ep$ denote the put option price given by \eqref{Pep}, and let $P_0$ and $P_1$ be as defined in \eqref{P0} and \eqref{P1}, respectively. Then, under assumptions \ref{assmpn}, \[|P_\ep(\tau,x,y)-(P_0(\tau,x)+\sqrt{\ep}P_1(\tau,x))|\leq O(-\ep\log\ep).\] 
\end{thm}
\begin{proof}
We use a maximum/comparison principle argument to prove this result.

Define \begin{equation}\label{psi}\psi_\ep(y)=\ep (y-m)^2+D_\ep,
\end{equation}
where $D_\ep>0$ is  a constant chosen such that
\begin{equation}\label{psi-ineq}
\frac{\psi_\ep(y)}{2}>|u_2(\tau,x,y)|+\sqrt{\ep}\ |u_3(\tau,x,y)|\end{equation} for all  $(\tau,x,y)\in [0,T]\times\R\times\R.$ It is possible to find such a $D_\ep$, as $U_2$ and $u_3$ have at most logarithmic growth in $y$ (see \eqref{u0}, \eqref{u1-formula}, \eqref{U2}, \eqref{u3-equation} and Lemma \ref{bounded-phi-prime}) and are bounded functions of $\tau$ and $x$ (see \eqref{U2} and \eqref{u0}).  It is easy to check that $D_\ep=O(-\log \ep)$, see Lemma \ref{D_ep} in the Appendix.

Choose a constant $C>0$  large enough so that 
\begin{equation}\label{const-C}
C+2(y-m)^2>\p_\tau U_2-\hat{\L}_0U_2-\L_1 u_3.
\end{equation}
It is possible to choose such a $C$ since $U_2$ and $u_3$ are bounded functions of $\tau$ and $x$ and have at most logarithmic growth in $y$ (see \eqref{u0}, \eqref{u1-formula}, \eqref{U2},  \eqref{u3-equation} and Lemma \ref{bounded-phi-prime}).

Define \[\ul{u}_\ep(\tau,x,y):=u_0(\tau,x)+\sqrt{\ep} u_1(\tau, x)+\ep U_2(\tau,x,y)+\ep^{3/2}u_3(\tau,x,y)-\ep \psi_\ep(y)-\ep (C+\bar{c}_2^2)\tau,\]
and 
\[\ol{u}_\ep(\tau,x,y):= u_0(\tau,x)+\sqrt{\ep} u_1(\tau, x)+\ep U_2(\tau,x,y)+\ep^{3/2}u_3(\tau,x,y)+\ep \psi_\ep(y)+\ep (C+\bar{c}_2^2)\tau, \]
where $\bar{c}_2$ is the upper bound on $\s_2$ in Assumption \ref{assmpn}, $u_0, u_1, U_2$ and $\psi_\ep$ are defined in \eqref{u0}, \eqref{u1-formula}, \eqref{U2} and  \eqref{psi} respectively. 

We compute
\begin{align*}
\p_\tau \ul{u}_\ep-\L_\ep\ul{u}_\ep=&\p_\tau u_0+\sqrt{\ep}\p_\tau u_1+\ep\p_\tau U_2+\ep^{3/2}\p_\tau u_3-\ep(C+\bar{c}_2^2)\\
&-\L_0u_0-\sqrt{\ep}\hat{\L}_0 u_1-\ep\hat{\L}_0U_2-\ep^{3/2}\hat{\L}_0u_3\\
&-\sqrt{\ep}\L_1 U_2+\sqrt{\ep}\L_1\psi_\ep-\ep\L_1 u_3\\
&-\B U_2-\sqrt{\ep}\B u_3+\B\psi_\ep-\ep\frac{\gamma(1-\rho^2)\s^2_1(y)}{2}\l(\p_y U_2+\sqrt{\ep}\p_y u_3-\p_y\psi_\ep\r)^2,
\end{align*}
where we have used the independence of $y$ in $u_0$ and $u_1$ to eliminate some terms. Using Remark \ref{Rem-U2}, \eqref{Poisson-u2},  \eqref{u3-equation} and $\B \psi_\ep< -\ep 2(y-m)^2+\ep\bar{c}_2^2$, we get
\begin{align*}
\p_\tau \ul{u}_\ep-\L_\ep\ul{u}_\ep
<& -\ep(C+2(y-m)^2)\\
&+\ep^{3/2}\l[-2(\rho+\eta\sqrt{1-\rho^2})\frac{b(y)\s_2(y)}{\s_1(y)}(y-m)+\p_\tau u_3-\hat{\L}_0u_3\r]\\
&+\ep\Bigl[\p_\tau U_2-\hat{\L}_0U_2-\L_1u_3-\frac{\gamma(1-\rho^2)\s^2_1(y)}{2}\l(\p_y U_2+\sqrt{\ep}\p_y u_3-\ep 2(y-m)\r)^2\Bigr].
\end{align*}
The term $-\ep(C+2(y-m)^2)$ dominates the other terms on the right hand side of the above equation, uniformly, for small enough $\ep$. To see this, first look at the terms of  $O(\ep^{3/2})$. The first term grows linearly in $y$ as $b$ and $\s_2$ are bounded and $\s_1$ is bounded away from $0$. The terms that depend on $u_3$ grow at most logarithmically in $y$ and are bounded in $\tau$ and $x$. This can be seen from \eqref{u0}, \eqref{u1-formula}, \eqref{u3-equation}, \eqref{u2} and Lemma \ref{bounded-phi-prime}. Thus the $O(\ep^{3/2})$ terms are dominated by  $\ep(C+2(y-m)^2)$ uniformly for all $\tau,x,y$, for small enough $\ep$.  We now turn to the remaining term of $O(\ep)$ in the third line of the above inequality. By choice of $C$ in \eqref{const-C}, this term is dominated by  $\ep(C+2(y-m)^2)$.
 Therefore the term $-\ep(C+2(y-m)^2)$ dominates the other terms  for small enough $\ep$ making the right hand side of the above inequality negative. Hence,
 \begin{equation}\label{max-principle-ineq1}
\p_\tau \ul{u}_\ep(\tau,x,y)-\L_\ep\ul{u}_\ep(\tau,x,y)<0, \quad \forall (\tau,x,y)\in (0,T)\times \mathbb R\times \mathbb R,
\end{equation}
for small enough $\ep$.
Similarly, it can be shown that \begin{equation}\label{max-principle-ineq2}
\p_\tau \ol{u}_\ep(\tau,x,y)-\L_\ep\ol{u}_\ep(\tau,x,y)>0, \quad \forall (\tau,x,y)\in (0,T)\times \mathbb R\times \mathbb R,
\end{equation}
for small enough $\ep$.

By the maximum principle, we can show that $\ul{u}_\ep\leq u_\ep\leq \ol{u}_\ep$, as follows. 
The goal is to prove that at the point of maximum of $\ul{u}_\ep-u_\ep$ we have $\ul{u}_\ep\leq u_\ep$. For this  we will first need to ensure the maximum of  $\ul{u}_\ep-u_\ep$ is attained and to this end we first perturb our function $\ul{u}_\ep$ slightly.
 Define \[\ul{u}_\ep^\delta=\ul{u}_\ep-\delta\l(\sqrt{1+x^2}+\frac{1}{T-\tau}\r),\] where $0<\delta\ll1$. Note that the first and second order derivatives of $\sqrt{1+x^2}$ are bounded.  Therefore the strict inequalities \eqref{max-principle-ineq1} and \eqref{max-principle-ineq2} still hold if $\underline{u}_\ep$ is replaced with $\underline{u}_\ep^\delta$, for small enough $\delta$.  
Because of the  term  $-\psi_\ep(y)$ in $\ul{u}_\ep$  which dominates the growth in $y$, $\ul{u}_\ep\to-\infty$ when $|y|\to\infty$. Also observe that $\ul{u}_\ep$ is bounded in $\tau$ and $x$, so the perturbation of $\ul{u}_\ep$ by $-\delta\l(\sqrt{1+x^2}+\frac{1}{T-\tau}\r)$ ensures $\ul{u}^\delta_\ep\to -\infty$ when either  $|x|\to\infty$ or $\tau\to T$. Since $u_\ep$ is bounded, we see that $\ul{u}^\delta_\ep-u_\ep\to -\infty$ when either $|y|\to\infty$, $|x|\to\infty$ or $\tau\to T$. Therefore, $\ul{u}^\delta_\ep-u_\ep$ must attain its maximum at some finite point, say $(\tau_0,x_0,y_0)$, in the domain $[0,T)\times\R\times\R$. For  $(\tau_0,x_0,y_0)$  in the interior of $(0,T)\times \mathbb R\times \mathbb R$, we have \begin{equation}\label{max-principle}0=\p_\tau u_\ep(\tau_0,x_0,y_0)-\L_\ep u_\ep(\tau_0,x_0,y_0)\leq \p_\tau\ul{u}^\delta_\ep(\tau_0,x_0,y_0)- \L_\ep\ul{u}^\delta_\ep(\tau_0,x_0,y_0)<0;\end{equation} which is a contradiction. Therefore the point of maximum must occur  at the boundary $\tau=0$. Recall that at $\tau=0$, $u_0(0,x)=u_\ep(0,x,y)=\ -[K-Ke^x]^+$ and $u_1(0,x)=0$ from \eqref{u1-formula}.  Then, by the choice of $D_\ep$ in the definition of $\psi_\ep$ i.e.\ \eqref{psi-ineq}, we get $\ul{u}^\delta_\ep(0,x,y) <u_\ep(0,x,y)$, for all $(x,y)\in\R\times\R$ and thus $\ul{u}^\delta_\ep<u_\ep$ everywhere, for all $\delta\ll1$. Taking the limit at $\delta\to 0$, we have the desired comparison \begin{equation}\label{ineq1}\ul{u}_\ep(\tau,x,y)\leq u_\ep(\tau,x,y), \quad \forall (\tau,x,y)\in [0,T]\times \mathbb R\times \mathbb R.\end{equation}
 By a similar argument, we get \begin{equation}\label{ineq2}
\ol{u}_\ep\geq u_\ep(\tau,x,y), \quad \forall (\tau,x,y)\in [0,T]\times \mathbb R\times \mathbb R.
\end{equation}

Putting inequalities \eqref{ineq1} and \eqref{ineq2} together and using \eqref{psi-ineq}, we get
\[ |u_\ep(t,x,y)-(u_0(\tau,x)+\sqrt{\ep}u_1(\tau,x))|\leq \frac{\ep\psi_\ep(y)}{2}+\ep(C+\bar{c}_2^2)\tau,\] for all $(\tau,x,y)\in [0,T]\times \mathbb R\times \mathbb R$.  Recall that the constant $D_\ep$ in the definition of $\psi_\ep$ is of order $-\log(\ep)$ and that $\psi_\ep$ has a term with quadratic growth in $y$, which gives us 
\[|u_\ep(t,x,y)-(u_0(\tau,x)+\sqrt{\ep}u_1(\tau,x))|\leq O(-\ep \log\ep)+O(\ep^2)y^2.\]

We can repeat the same argument for $\tilde{u}_\ep$, which leads to the desired result
\[ |P_\ep(t,x,y)-(P_0(\tau,x)+\sqrt{\ep}P_1(\tau,x))|\leq O(-\ep \log\ep)+O(\ep^2)y^2,\] for all $(\tau,x,y)\in [0,T]\times \mathbb R\times \mathbb R$. 
\end{proof}

\section{Implied Volatility}\label{sec:IV}
In theory, we could use the corrected option price formula for calibration. However, option prices are typically quoted in terms of their implied volatility and so it is useful to derive a corrected implied volatility formula corresponding to the corrected option price.  Let $I_\ep$ denote the implied volatility corresponding to the put option price $P_\ep$. Recall that $x=\ln(S/K)$ and  the Black-Scholes formula for put option prices with volatility $\s$ is given by the formula
\[P_{BS}(\tau,x;\s)=KN(-d_2)-Ke^xN(-d_1),\]
where $N(z)=\frac{1}{\sqrt{2\pi}}\int_{-\infty}^z e^{-y^2/2}dy$, $d_1=\frac{x+\frac1 2 \s^2\tau}{\s\sqrt{\tau}}$ and 
$d_2=\frac{x-\frac1 2 \s^2\tau}{\s\sqrt{\tau}}$. 

By definition, the implied volatility $I_\ep$ is obtained by setting 
\begin{equation}\label{IV-definition}P_{BS}(\tau,x;I_\ep(\tau,x,y))=P_\ep(\tau,x,y).\end{equation}
From  \eqref{P0}, we know that the leading order term of the put option price $P_0(\tau,x)=P_{BS}(\tau,x;\overline{\s_1})$. 
 Thus, the leading order term of the implied volatility is simply $\overline{\s_1}$.  To obtain a correction term for the asymptotic implied volatility we will expand $I_\ep$ about $\overline{\s_1}$ in powers of $\sqrt{\ep}$ as follows \begin{equation}\label{imp-vol-expansion}
I_\ep(\tau,x,y)=\overline{\s_1}+\sqrt{\ep}I_1(\tau,x,y)+\ep I_2(\tau,x,y)+\cdots.\end{equation}
Using the expansion of $I_\ep$ in \eqref{IV-definition}, we get
\begin{equation}\label{equating}
\begin{split}
&P_{BS}(\tau,x;\overline{\s_1})+\sqrt{\ep}I_1\frac{\p P_{BS}}{\p\s}(\tau,x;\overline{\s_1})+\cdots\\
&=P_0(\tau,x)+\sqrt{\ep}P_1(\tau,x,y)+\cdots.
\end{split}
\end{equation} 
Thus the correction term to the implied volatility is 
\begin{equation}\label{I1}
I_1=P_1(\tau,x,y)\l[\frac{\p P_{BS}}{\p\s}(\tau,x;\overline{\s_1})\r]^{-1}.
\end{equation}

Differentiating the Black-Scholes formula, we get
\[\frac{\p P_{BS}}{\p\s}=\frac{Ke^{-d_2^2/2}\sqrt{\tau}}{\sqrt{2\pi}}.\]
Substituting the formula for $P_1$ in \eqref{I1}, we can rewrite \eqref{imp-vol-expansion}  as 
\begin{equation}\label{Iep}\begin{split}I_\ep=&\overline{\s_1}-\sqrt{\ep}\frac{\sqrt{2\pi}}{K\sqrt{\tau}e^{-d_2^2/2}}\bigl[\tau\bigl( -A\p^3_{xxx}P_0(\tau,x)+(A+B)\p^2_{xx}P_0(\tau,x)\\
&-B\p_xP_0(\tau,x)\bigr) \bigr]+o(\sqrt{\ep}),\end{split}\end{equation}
 where $A$ and $B$ are defined in \eqref{A} and \eqref{B} respectively.
 
The constants $A$ and $B$ can be simplified to:
\begin{equation}\label{A-alt}
A=\int_{-\infty}^\infty\rho\frac{\s_1(y)}{\s_2(y)}\l(\int_{-\infty}^y (\s_1^2-\overline{\s_1}^2)\pi(dy)\r)dy,
\end{equation}
and
\begin{equation}\label{B-alt}
B=\int_{-\infty}^\infty(\rho+\eta\sqrt{1-\rho^2})\frac{b(y)}{\s_1(y)\s_2(y)}\l(\int_{-\infty}^y (\s_1^2-\overline{\s_1}^2)\pi(dy)\r)dy.
\end{equation}
On substituting the derivatives of $P_0$ in \eqref{Iep}, we get
\begin{equation}
\begin{split}
I_\ep&=\overline{\s_1}+\sqrt{\ep}\l[ A\frac{d_2}{\overline{\s_1}^2\sqrt{\tau}}+\frac{B}{\overline{\s_1}}\r]+o(\sqrt{\ep})\\
&=\overline{\s_1}+\sqrt{\ep}\l[A\frac{x}{\overline{\s_1}^3\tau}+\frac{B-A/2}{\overline{\s_1}}\r]+o(\sqrt{\ep}).
\end{split}
\end{equation}
We see that, as in the no-arbitrage pricing case (see \cite{FPS00}),  the corrected implied volatility is an affine function of log-moneyness-to-maturity ratio (LMMR), where
\[LMMR=\frac{\ln(strike\ price/asset\ price)}{time\ to\ maturity}=\frac{-x}{\tau}.\]

For calibration purposes, it is convenient to rewrite the formula for implied volatility, $I_\ep$, as 
\begin{equation}\label{IV-surface}
I_\ep=a(LMMR)+d+o(\sqrt{\ep}),\end{equation}
where
$a=\frac{-\sqrt{\ep}A}{\overline{\s_1}^3}$, and 
$d=\overline{\s_1}+\frac{\sqrt{\ep}}{\overline{\s_1}}\l(B-\frac{A}{2}\r)$. On calibration of $a$ and $d$ from market implied volatility data, we can determine $A$ and $B$ from the formulas 
$A=\frac{-\overline{\s_1}^3 a}{\sqrt{\ep}}$ and $B=\frac{1}{\sqrt{\ep}}\l((d-\overline{\s_1})\overline{\s_1}-\frac{\overline{\s_1}^3 a}{2} \r)$.

In \cite{FPS00} an affine function of LMMR is fitted to S\&P 500 European call option implied volatility data and $a$ and $d$ are estimated to be -0.154 and 0.149 respectively. Using these estimates we graph the implied volatility surface given by \eqref{IV-surface}, see Figure \ref{fig1}.

\begin{figure}[p]
  \centering
      \includegraphics[width=4.5in]{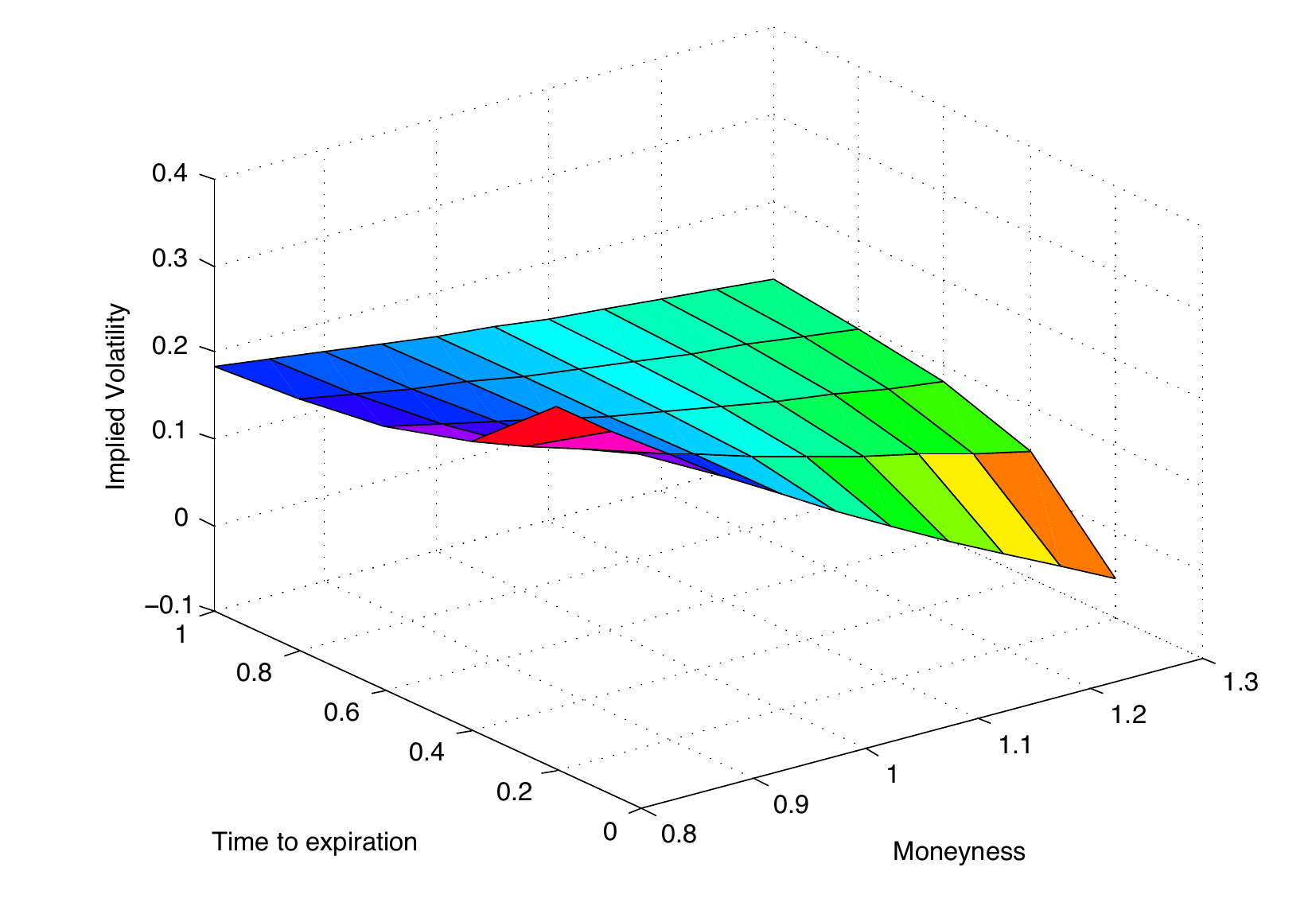}
  \caption{Implied volatility surface given by \eqref{IV-surface} with $a=-0.154$ and $d=0.149$.}
  \label{fig1}
\end{figure}

Taking $Y_t$ to be an OU-process and $\sigma_1(y)$ an arctangent volatility function, we compute the implied volatility from the corrected asymptotic option price formula in Theorem 1. To be precise, we take $\sigma_1(y)=0.3+\frac{0.5}{\pi}\arctan(y)$, $\sigma_2(y)=0.2$, $m=0$, $b(y)=1$ and $\rho=-0.2$ in our stochastic volatility model \eqref{FSDEs}. We fix $\tau=0.25$ and assume $\ep=0.004$.  Figure \ref{fig2} gives the implied volatility as a function of log moneyness i.e.\ $-x=\log K/S$, for three different values of the risk parameter $\eta$. The skew of the implied volatility function is clearly seen. Observe that as the value of $\eta$ increases, the implied volatility gets shifted down. This was also observed in \cite{SS} (see Figure 3 in \cite{SS}). In Figure \ref{fig2} implied volatility appears to be an almost  linearly decreasing function of log moneyness. This agrees with the formally derived formula for corrected  asymptotic implied volatility in \eqref{IV-surface}.

\begin{rem} The corrected implied volatility formula only depends on the volatility risk parameter $\eta$ and the correlation term $\rho$, and has no dependence on the risk aversion parameter $\gamma$. By calibration,  the correlation coefficient $\rho$ and the risk parameter $\eta$ can be obtained from \eqref{A-alt} and \eqref{B-alt}.
\end{rem}

\begin{figure}[p]
  \centering
      \includegraphics[width=4in]{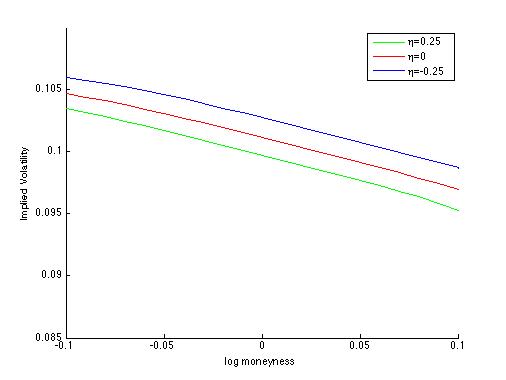}
  \caption{Implied volatility as a function of log moneyness for the arctangent stochastic volatility model, when $\eta=-0.25,0,0.25$.}
  \label{fig2}
\end{figure}


\appendix

\section{Appendix}

\setcounter{lemma}{0}
    \renewcommand{\thelemma}{\Alph{section}.\arabic{lemma}}
    \setcounter{equation}{0}
    \numberwithin{equation}{section}
 
\begin{lemma}\label{bounded-phi-prime}
Suppose $f\in C_b(\R)$ is a bounded continuous function which is centered with respect to the invariant distribution $\pi$, i.e. $\int fd\pi=0$. Then,  the Poisson equation \[\B v=f,\] has a solution that grows at most logarithmically and has bounded derivative $v^\prime$.
\end{lemma}
\begin{proof}
We construct a solution that satisfies the required growth condition.

Suppose $f$ is bounded and $v$ is a solution of $\B v=f$. Recall the definition of the differential operator $\B$ in \eqref{generator-B}. Multiplying the equation, $\B v=f$, by the integrating factor $\exp\{\int_m^y\frac{2(m-z)}{\s_2^2(z)}dz\}$, we get
\begin{equation*}
\l(e^{\int_m^y\frac{2(m-z)}{\s_2^2(z)}dz}v^\prime(y)\r)^\prime=\frac{2f(y)}{\s_2^2(y)}e^{\int_m^y\frac{2(m-z)}{\s_2^2(z)}dz}
\end{equation*}
Without loss of generality, we can assume $m=0$, so
\begin{equation*}
\l(e^{\int_0^y\frac{-2z}{\s_2^2(z)}dz}v^\prime(y)\r)^\prime=\frac{2f(y)}{\s_2^2(y)}e^{\int_0^y\frac{-2z}{\s_2^2(z)}dz}
\end{equation*}
Integrating, we get
\begin{equation}\label{soln}
v^\prime(y)=e^{\int_0^y\frac{2z}{\s_2^2(z)}dz}\int_{-\infty}^y \frac{2f(u)}{\s_2^2(u)}e^{\int_0^u\frac{-2z}{\s_2^2(z)}dz}du,
\end{equation}
under the assumption that $\lim_{y\to-\infty}e^{\int_0^y\frac{-2z}{\s_2^2(z)}dz}v^\prime(y)$ is $0$. Let us denote the right hand side of \eqref{soln} by the function $G(y)$ i.e.
\[G(y):=e^{\int_0^y\frac{2z}{\s_2^2(z)}dz}\int_{-\infty}^y \frac{2f(u)}{\s_2^2(u)}e^{\int_0^u\frac{-2z}{\s_2^2(z)}dz}du.\]
Using l'h\^{o}pital's rule as $|y|\to\infty$ and the boundedness of $f$ and $\s_2$, we see that $G$ is a bounded function. Then $v(y):=\int_0^yG(u)du$ gives us a solution to the Poisson equation $\B v=f$. We will show that this solution $v$ to the Poisson equation has at most logarithmic growth.

Recall the invariant measure $\pi$ given in \eqref{inv-meas}. Using the centering condition $\int fd\pi$, we get
\begin{equation*}
v^\prime(y)=e^{\int_0^y\frac{2z}{\s_2^2(z)}dz}\int_y^\infty \frac{2f(u)}{\s_2^2(u)}e^{\int_0^u\frac{-2z}{\s_2^2(z)}dz}du.
\end{equation*}
Let $y>1$. By the boundedness of $f$ and $\s_2$, we can bound
\begin{equation}\label{bd1}\begin{split}
|v^\prime(y)|&\leq ce^{\int_0^y\frac{2z}{\s_2^2(z)}dz}\int_y^\infty \frac{2}{\s_2^2(u)}e^{\int_0^u\frac{-2z}{\s_2^2(z)}dz}du\intertext{(for some constant $c>0$)}
&=ce^{\int_0^y\frac{2z}{\s_2^2(z)}dz}\int_y^\infty \frac{1}{-u}\frac{-2u}{\s_2^2(u)}e^{\int_0^u\frac{-2z}{\s_2^2(z)}dz}du,\intertext{which on integrating by parts gives}
&=ce^{\int_0^y\frac{2z}{\s_2^2(z)}dz}\l[\frac{1}{y}e^{-\int_0^y\frac{2z}{\s_2^2(z)}dz}-\int_y^\infty \frac{1}{u^2}e^{\int_0^u\frac{-2z}{\s_2^2(z)}dz}du \r]\\
&\leq \frac{c_1}{y} 
\end{split}
\end{equation}
for some $c_1>0$.

Let $y<-1$. We repeat the same argument using the bound
\begin{equation}\label{bd2}\begin{split}
|v^\prime(y)|&\leq e^{\int_0^y\frac{2z}{\s_2^2(z)}dz}\int_{-\infty}^y \frac{2|f(u)|}{\s_2^2(u)}e^{\int_0^u\frac{-2z}{\s_2^2(z)}dz}du\\
&=ce^{\int_0^y\frac{2z}{\s_2^2(z)}dz}\l[\frac{1}{-y}e^{-\int_0^y\frac{2z}{\s_2^2(z)}dz}-\int_{-\infty}^y \frac{1}{u^2}e^{\int_0^u\frac{-2z}{\s_2^2(z)}dz}du \r]\\
&\leq \frac{c_2}{-y}
\end{split}\end{equation}
for some $c_2>0$.
The boundedness of  $|v^\prime(y)|$ for all $y$ together with the bounds \eqref{bd1} and \eqref{bd2} when $|y|>1$ gives us \[|v(y)|\leq C_1\log(1+|y|)+C_2, \quad \forall y\in\R,\] for some $C_1,C_2>0$.
\end{proof}

\begin{lemma}\label{D_ep}
$D_\ep=O(-\log \ep)$.
\end{lemma}
\begin{proof}
Recall that $D_\ep$ is chosen to be a positive number such that \[D_\ep>|u_2(\tau,x,y)|+\sqrt{\ep}\ |u_3(\tau,x,y)|-\ep(y-m)^2,\] for all $\tau,x,y$.  
Since $u_2$ and $u_3$ are bounded in $\tau$ and $x$ and have at most logarithmic growth in $y$, we can write
\[|u_2(\tau,x,y)|+\sqrt{\ep}\ |u_3(\tau,x,y)|\leq c_1\log(1+(y-m)^2)+c_2\quad \text{ for all }\tau,x,y,\] for some positive constants $c_1$ and $c_2$. 
So it suffices to choose $D_\ep$ such that
 \begin{equation}\label{ineq}D_\ep>c_1\log(1+(y-m)^2)+c_2-\ep(y-m)^2,\end{equation} for all $\tau,x,y$.  
The quadratic term $(y-m)^2$ grows faster than $\log(1+(y-m)^2)$ for large $|y|$, 
therefore the maximum of the right hand side in \eqref{ineq} is attained at finite values of $y$.   We will determine the maximum value of the right hand side of \eqref{ineq} and choose $D_\ep$ to be larger than that. 

Let $y_0$ denote a point of maximum of the right hand side of \eqref{ineq}, then 
\[\frac{d}{dy}\l[ c_1\log(1+(y-m)^2)+c_2-\ep(y-m)^2 \r]|_{y=y_0}=0,\]
i.e.\[\frac{c_12(y_0-m)}{1+(y_0-m)^2}-\ep 2(y_0-m)=0.\]
Solving for $y_0$ we get, either $y_0=m$, in which case choose $D_\ep>c_2$ or 
\[(y_0-m)^2=\frac{c_1}{\ep}-1, \] and the right hand side of \eqref{ineq} becomes
\[c_1\log(\frac{c_1}{\ep})+c_2-c_1+\ep.  \] Therefore, for small enough $\ep$, it suffices to choose $D_\ep= -2c_1\log \ep$.
 
\end{proof}

\end{document}